\documentclass[a4paper,11pt,twoside]{article}

\usepackage{bbm,amsthm}
\usepackage{amsfonts}
\usepackage{amssymb}

\def\greaterthansquiggle{\raise.3ex\hbox{$>$\kern-.75em\lower1ex\hbox{$\sim$}}}
\def\lessthansquiggle{\raise.3ex\hbox{$<$\kern-.75em\lower1ex\hbox{$\sim$}}}
\newcommand{\beq}{\begin{equation}}
\newcommand{\eeq}{\end{equation}}
\newcommand{\beqa}{\begin{eqnarray}}
\newcommand{\eeqa}{\end{eqnarray}}
\newcommand{\beqan}{\begin{eqnarray*}}
\newcommand{\eeqan}{\end{eqnarray*}}
\newcommand{\ba}{\begin{array}}
\newcommand{\ea}{\end{array}}

\newcommand{\A}{{\cal A}}
\newcommand{\B}{{\cal B}}
\newcommand{\C}{{\cal C}}

\newcommand{\Ha}{{\cal H}}
\newcommand{\I}{{\cal I}}

\newcommand{\X}{{\cal X}}

\def\nz{\ifmmode {I\hskip -3pt N} \else {\hbox {$I\hskip -3pt N$}}\fi}
\def\zz{\ifmmode {Z\hskip -4.8pt Z} \else
       {\hbox {$Z\hskip -4.8pt Z$}}\fi}
\def\qz{\ifmmode {Q\hskip -5.0pt\vrule height6.0pt depth 0pt
       \hskip 6pt} \else {\hbox
       {$Q\hskip -5.0pt\vrule height6.0pt depth 0pt\hskip 6pt$}}\fi}
\def\rz{\ifmmode {I\hskip -3pt R} \else {\hbox {$I\hskip -3pt R$}}\fi}
\def\cz{\ifmmode {C\hskip -4.8pt\vrule height5.8pt\hskip 6.3pt} \else
       {\hbox {$C\hskip -4.8pt\vrule height5.8pt\hskip 6.3pt$}}\fi}
\newtheorem{theorem}{Theorem}
\newtheorem{definition}{Definition}
\newtheorem{lemma}{Lemma}

\def\au{{\setbox0=\hbox{\lower1.36775ex%
\hbox{''}\kern-.05em}\dp0=.36775ex\hskip0pt\box0}}
\def\ao{{}\kern-.10em\hbox{``}}

\normalsize
\begin{document}
\bibliographystyle{plain}
\begin{titlepage}



\begin{center}
{\Large \bf  Translation Invariant States on Twisted Algebras on a
Lattice}
\\[30pt]

Bernhard Baumgartner $^\ast$, Fabio Benatti $^{**}$,
Heide Narnhofer  $^\ast $\\
[10pt] {\small\it}
$^\ast$ Fakult\"at f\"ur Physik, Universit\"at Wien\\
$^{**}$ Department of Theoretical Physics, University of Trieste
\vfill \vspace{0.4cm}

\begin{abstract}
We construct an algebra with twisted commutation relations and
equip it with the shift. For appropriate irregularity of the non-local
commutation relations we prove that the
tracial state is the only translation-invariant state.
\vskip 2cm

\\
\hspace{1.9cm}

\end{abstract}
\end{center}

\vfill {\footnotesize}

email: bernhard.baumgartner@univie.ac.at

email: benatti@ts.infn.it

email: heide.narnhofer@univie.ac.at
\end{titlepage}

\section{Introduction}

Whereas the dynamics of quantum systems with finitely many degrees of
freedom is fully controllable via the spectrum of the unitary operator
that implements it, this is not longer true with infinitely many degrees
of freedom. While these systems can always be described by
operators acting on suitable Hilbert spaces, it may however happen,
typically in the thermodynamic
limit of interacting quantum systems,  that
not all of these operators represent observable quantities.
In such cases, what one does is to focus not upon
the dynamics, rather on its equilibrium (time-invariant) states.

A great step forward in our understanding of thermodynamics, in particular of
its second law, would follow by showing that for interacting systems the class
of invariant states is drastically reduced in comparison to those of
quasifree evolutions, where e.g. every state in equilibrium with
respect to any (space-)translation invariant dynamics is also (space-)
translation invariant. For interacting systems, one hopes that invariant
states might be characterized by only a few thermodynamic quantities.

Unfortunately, scarcely any control on the dynamics of
realistic physical systems is available; however, some steps forward
can be made following a toy model approach and constructing  algebras
and discrete dynamics on these algebras such that only one invariant
state exists. A
particular instance of invariant state is the so-called tracial, or
totally mixed, state $\omega$ whose two-point correlation functions
are such that $\omega(AB)=\omega(BA)$ for all $A,B$ in
the algebra of operators.

Examples of quantum dynamical systems where the tracial state is the only invariant state are the
Price-Powers shifts \cite{RTP,GLP},\cite{SNT}, \cite{NT} or the irrational
rotation algebras \cite{BNS},\cite{NT}. In both cases, what forbids
the existence of invariant states different from the tracial one is
that operators in the course of time anticommute infinitely often and
sufficiently irregularly with one another.
Such a lack of asymptotic commutativity (abelianess) is indeed
expected in real interacting quantum systems.

In the following, we give another example of such quantum dynamical
systems inspired very much by the Price-Powers shifts.
In this latter case, the algebra is created by self-adjoint operators
$e_k=e^*_k$, $e_k^2=1$, $k\in \mathbb{N}$, that commute or
anticommute, $e_ke_p=(-1)^{g(|p-k|)}e_pe_k$,
as prescribed by a so-called \textit{bitstream}, namely a two-valued
function on the integers,
$g:\mathbb{N} \mapsto\{0,1\}$, $g(0)=0$.
In the following, we will consider Weyl-operators in place of the
$e_k$ and organize them in such a way that
shifted Weyl-operators remain in Weyl-like relation with one another
in a so-called \textit{complementary} manner \cite{DP} so that any two
of them create a full matrix algebra $M_{d\times d}$.

We shall show that, under the hypothesis of sufficiently
irregular complementary Weyl-like relations,
the tracial state is the only translationally
invariant state as for the Powers-Price shifts.
The paper is organized as follows:
\begin{itemize}
\item
in chapter 2. we define the algebra and the complementary relations,
together with some representations either as an AF-algebra
(in some cases a UHF algebra) or as a
quantum-spin chain by means of generalized Jordan-Wigner transformations;
\item
in chapter 3. we show that, under the assumption of sufficiently
random commutation relations, only the tracial state can be
translation invariant.
\end{itemize}

\section{The Algebra and its Automorphism}

\subsection{Definitions}\label{definitions}

We start with infinitely many finite-dimensional algebras $\A_m$,
$m\in \mathbb{N}$,  all isomorphic to $d\times d$ matrix algebras,
created by the operators $W^{(m)} _{\vec k}$,
$\vec{k}\in\mathbb{Z}^2_d:=\{(k_1,k_2),\ k_i=0,1,\ldots,d-1\}$,
that satisfy the commutation relations of a discrete Weyl group

\beq
\label{Weyllett}
W^{(m)} _{\vec{k}_1}W^{(m)}_{\vec{k}_2}=
e^{\frac{2\pi i\sigma(\vec{k}_1,\vec{k}_2)}{d}}
W^{(m)}_{\vec{k}_2}W^{(m)} _{\vec{k}_1}\ ,
\eeq
with symplectic form $\sigma(\vec{k}_1,\vec{k}_2):=k_{11}k_{22}-k_{12}k_{21}$.

Note that $W^{(m)}_{\vec{0}}=\mathbbm 1$, and  $W^{(m)}_{-\vec{k}}= (W^{(m)}_{\vec{k}})^{-1}$.

The relations between Weyl operators with different upper indices
are \textit{twisted} by means of a sequence of $2\times 2$ matrices
$A_n$, $n\in \mathbb{Z}$, with entries in $\{0,1,\ldots, d-1\}$;
explicitly,
\beqa
\label{twistedWeyl}
W^{(p)}_{\vec{k}_p}W^{(q)} _{\vec{k}_q} &=&
{\rm e}^{2\pi\,i\,u_{\vec{k}_p\vec{k}_q}(q-p)}
\, W^{(q)}_{\vec{k}_q}W^{(p)} _{\vec{k}_p}\ ,
\\
\label{twistedWeyl1a}
u_{\vec{k}_p\vec{k}_q}(q-p)&:=&
\frac{1}{d}\,\sigma(\vec{k}_p,
A_{q-p}\vec{k}_q)\ .
\eeqa

Setting $A_0=\mathbbm 1$, the single-site relations (\ref{Weyllett})
are a particular instance of (\ref{twistedWeyl}).

The finite products define elements of an infinite discrete group. We denote them as
\beq
\label{Weylword} W_I:=W^{(1)}_{\vec{k}_1}W^{(2)}_{\vec{k}_2}\cdots
W^{(\ell)}_{\vec{k}_\ell}\ ,
\eeq
where $I$ denotes
a sequence of vectors $\{\vec{k}_m\}_{m\in\mathbb{Z}}$ with only finitely
many components, $(\vec{k}_1,\vec{k}_2,\ldots,\vec{k}_\ell)$, possibly
different from the vector $\vec{0}$.
We define a star operation as $(W^{(m)}_{\vec{k}})^*= (W^{(m)}_{\vec{k}})^{-1}$,
and $(UV)^*=V^*U^*$ as usual.
So the Weyl operators and their products are unitary elements of
the $C^*$ algebra generated by the finite products $W_I$,
which are assumed to be linearly independent.
We shall denote this algebra by $\A$.

We shall further equip $\A$ with an automorphism
$\alpha:\A\mapsto\A$ such that
$\alpha^n\left(W^{(0)}_{\vec{k}}\right)=W^{(n)} _{\vec{k}}$.
Then, generic algebraic
relations read
\beqa
W_I\alpha^n(W_J)&=&\Bigl(\prod_{a=1}^{n_I}W^{(a)}_{\vec{k}_a}\Bigr)\,
\Bigl(\prod_{b=1}^{n_J}W^{(b+n)}_{\vec{k}_b}\Bigr)=
{\rm e}^{2\pi\,i\,u_n(I;J)}\,
\alpha^n(W_J)\,W_I\ ,
\\
\label{gen-alg-rel}
u_n(I;J)&:=&
\label{comrel}
\frac{1}{d}\sum_{a=1}^{n_I}
\sum_{b=1}^{n_J}\sigma(\vec{k}_a\,,\,A_{b+n-a}\vec{k}_b)=
\sum_{a=1}^{n_I}\sum_{b=1}^{n_J}u_{\vec{k}_a\vec{k}_b}(b+n-a)\ .
\eeqa
\medskip

\noindent
\textbf{Remarks 1}\quad

\begin{enumerate}
\item
Every sequence of $2\times 2$ matrices $\{A_n\}_{n\in\mathbb{N}}$
generates its own algebra $\A$; however, for special sequences
$\{A_n\}_{n\in\mathbb{N}}$
(as will be shown in the next section), the
corresponding algebras built by $W^{(m)} _{\vec k}, m=1,,,l$ will be isomorphic $\forall l$ to the tensor product over $l$ local sizes and therefore the total algebra can be
considered as the same $\A$ equipped with the usual shift and a different automorphism
$\alpha$ with finite speed.
\item
We will refer to the Weyl
operators $ W^{(m)} _{\vec{k}}$ in (\ref{Weyllett}) as to the letters of
the algebra $\A$ and to the products $W_I$ as in (\ref{Weylword}) as to the
words of $\A$.
The norm of every word equals
$1$ and $(W_I)^d=\mathbbm 1$, as the eigenvalues of all $W_I$ are the pure phases
$e^{\frac{2\pi i\ell}{d}}$, $0\leq\ell\leq d-1$.
Linear combinations of words will be referred to as sentences,
as in the case of the Price-Powers-shift. A state
$\omega:\A\mapsto\C$ over the algebra amounts to a positive,
normalized functional over $\A$: it is thus fixed by
giving its values on all words.
\end{enumerate}

We proceed in studying the algebras,
showing the existence of a non-trivial representation.
This is not a trivial problem, see \cite{AP}.
Consider the set $\I$ of multiindices $I$ and define the
composition law of addition: $\I\times\I\mapsto\I$:
if $I=\{\vec{k}_m^I\}_{m\in\mathbb{Z}}$ and
$J=\{\vec{k}_m^J\}_{m\in\mathbb{Z}}$, then
$$
(I,J)\mapsto I+J=\{\vec{k}^I_m+\vec{k}^J_m\}_{m\in\mathbb{Z}}\ ,
$$
where the sum of vectors
$\vec{k}^I_m+\vec{k}^J_m$ is understood modulo $d$.
Then, the family of operators
$\{\widetilde{W}_I\}_{I\in \I}$ satisfying the multiplication law
$\widetilde{W}_I\,\widetilde{W}_J=\widetilde{W}_{I+ J}$, form an
Abelian group $G$ on which the shift defines an automorphism with an
associated shift-invariant measure $\delta$ such that
$\delta(\widetilde{W}_I)=0$ unless $I=I_0:=\{\vec{0}\}_{m\in\mathbb{Z}}$.
Therefore, one can consider the Hilbert space $\ell_2(G)$ spanned by
the orthonormal elements $\widetilde{W}_I$ and represent the Weyl
operators $W_I$ introduced before by
\beq
\label{cocycle}
\Pi(W_I)
\widetilde{W}_J=\underbrace{{\rm e}^{i\pi\,u_0(I,J)}}_{\omega(I;J)}
\,\widetilde{W}_{I+ J}\ .
\eeq
 $\omega(I;J)$ is a cocycle, namely
$$
\omega(I_1,I_2+I_3)\,\omega(I_2,I_3)=\omega(I_1+I_2,I_3)\,\omega(I_1,I_3)
\ ;
$$
In this way $\A$ is represented as a sub-algebra
$\Pi(\A)\subseteq\B(\ell_2(G))$ (it is known as the regular representation)
of the bounded operators on $\ell_2(G)$ and thus
all considerations in \cite{NS} are therefore applicable
to $\A$. For instance, $\A$ has a trivial center if
to any word $W_{I}$ there exists another word $W_{J}$
such that $W_{I}$ and $W_{J}$ do not commute. In this
case the trace on the
algebra
\begin{equation}
tr(W_{I})=0 \quad \forall I\neq \emptyset, \quad
tr(1)=1
\end{equation}
is unique and it is implemented in the regular
representation by
$$
tr(W_I)=\langle\widetilde{W}_{I_0}\vert\Pi(W_I)\vert\widetilde{W}_{I_0}\rangle
\ .
$$
Evidently. the trace is
invariant under the shift automorphism $\alpha:\A\mapsto\A$
\medskip

Like in \cite{NS}, the main interest is in finding conditions on the
cocycle (\ref{cocycle}) such that no other invariant states exist
other than the tracial state.
In \cite{NS} the main tool was the high degree of anticommutativity,
a generalization of the fact, that
translation invariant states over Fermi systems have to be even.
Though this criterion is sufficient only if $d=2$, non-commutativity
as embodied in (\ref{twistedWeyl}) will nevertheless turn out to be just as powerful in
restricting the class of invariant states. We will indeed give other
arguments to enlarge the class of automorphisms that allow only the
tracial state as invariant state. Though not optimal, the result
indicates that delocalization by the dynamics as we describe it in (2.2)is an effect which is
worth studying in more detail and in more realistic thermodynamic
systems.

\subsection{Spin-Chain Representation}\label{spinchainerep}

We have already given a representation of the algebra over $\B
(l_2(G)).$ However, in order to bring it in closer contact with
physical models we seek connections with spin chains and their automorphisms.

This demands different representations.
To demonstrate the differences we make a short detour,
considering finite algebras defined on a finite ring of $N$ lattice points (= upper indices)
with a cyclic shift, instead of an infinite set of points with a non-recurrent shift.
We get different dimensions of the Hilbert spaces. In the above mentioned representation
there are $(d^2)^N$ basis vectors.
For a spin-chain one would expect only $d^N$ as necessary.
But this, as will turn out, is possible for a restricted set of defining sequences only.
In Section \ref{jordanwigner} we give then a representation with a double-spin-chain,
possible for any defining sequence,
employing again $(d^2)^N$ basis vectors for the finite algebra on a ring.

More precisely, we shall try to represent the Weyl operators
$W^{(m)}_{\vec k}$, $m=0,\ldots,N$, as elements of the full matrix
algebra $\bigotimes_{n=0}^N(M_{d\times d})_n$.
We proceed step by step: let us define the Weyl operator at site
$0\leq j\leq N$ to be
\beq
\label{weyl-spin1}
W^{(j)}_{\vec k}=W_{A_{0,j}\vec{k}}\otimes
W_{A_{1,j}\vec{k}} \otimes W_{A_{j,j}\vec{k}}\otimes 1_{j+1}\ldots
\otimes 1_N\ .
\eeq
The unknowns in the construction are the $2\times 2$ matrices with
integer entries from $\{0,1,\ldots,d-1\}$ that we have to adjust in order
to fulfil the commutation relations (\ref{twistedWeyl}).
Therefore, from (\ref{Weyllett}), one gets the condition
\beq
\label{weyl-spin2}
\sum_{\ell=0}^j\sigma (A_{\ell,j}\vec k, A_{\ell,j}
\vec{m})=\sigma(\vec{k},\vec{m})\ ,
\eeq
forall $\vec{k}$, $\vec{m}$ and $0\leq j\leq N$
which is equivalent to
$\sum_{\ell=0}^j Det(A_{\ell,n})=1$.
If $d$ is prime, then all Weyl operators are unitarily isomorphic so
that the algebra created by them is $M_{d\times d}$; therefore, in the
rest of this section, $d$ will be assumed to be a prime number.

We have to control whether this ansatz can really be satisfied and how
far the matrices $A_{\ell,k}$ are determined by the matrices $A_n$
in (\ref{twistedWeyl1a}). It turns out that $A_{0,j}=A_j$ while $A_{00}=1$.
The other matrices
$A_{\ell,k}$ have to be calculated recursively from
(\ref{twistedWeyl})
and (\ref{twistedWeyl1a}). More
precisely to evaluate
\begin{equation}
\sigma (A_{0,1} \vec k, A_{0,2} \vec
l)+\sigma  (A_{1,1} \vec k,A_{1,2} \vec l) = \sigma (\vec
k,A_{0,1} \vec l)
\end{equation}
we define the linear map $A\mapsto \widehat{A}$,
$$
A=\left ( \ba{cc} a_{11} & a_{12} \\
a_{21} & a_{22} \ea \right ) \rightarrow   \widehat{A}:=
\left ( \ba{cc}
a_{22} & -a_{12} \\ a_{21} & a_{11} \ea \right )
$$
With this map, it follows that
$A_{1,2}=\hat A_{1,1}^{-1}(A_{0,1}-\hat A_{0,1}A_{0,2})$.
This fixes
$A_{1,2}$ if we take into account that the freedom in $A_{1,1}$ is
reduced to an isomorphism inside of the local algebra. However, we
can only be sure that there exists a solution $A_{1,2}$ if $A_{1,1}$ is
invertible which surely holds if $Det(A_{1,1})\neq 0$.
Similarly
$$
A_{n-k,n}=\hat{A}_{n-k,n-k}^{-1}(A_{0,n-k}-\sum _{l=0}^{n-k-1}\hat{A}_{l,n-k}A_{l,n})\
 .
$$
The equations are uniquely solvable (up to
trivial local isomorphisms) under the constraint
\beq
\label{cond-det}
Det(A_{0,j})+\ldots +Det(A_{j-1,j})\neq 1\quad \forall j\ .
\eeq
This implies in addition that $Det(A_{j,j})\neq 0$; as a consequence the algebra
generated by the Weyl operators (\ref{weyl-spin1}) is
isomorphic to the full matrix algebra.

Only the sequence of matrices $A_n$ is to our disposal, whereas
the matrices
$A_{\ell,k}$ are linearly depending on them; however, there remains
enough freedom to find sequences that meet the condition (\ref{cond-det}).
A special example corresponds to choosing
$A_{0,j}=\delta _{1,j}\mathbbm1$; this in turn corresponds to the usual shift on the lattice algebra.
More generically, one may choose $A_{0,j}$ to belong to a left ideal with determinant $0$ for all
$j\neq 0$;
it then follows that, for all $k>0$, $A_{n-k,n}$ also belongs to this ideal and therefore $Det(A_{n,n})=1 \quad \forall n.$

As a consequence, provided $Det(A_{n,n})\neq 0\quad \forall n $ we can consider the algebra to be the same algebra i.e. the
spin chain $M_{d\times d}^{\otimes\infty}$, but equipped with different automorphisms
corresponding to the different sequences $A_{0,n}.$ Notice that in the spin-chain representation, the algebra is fairly
simple, while the automorphism (which we will again denote by $\alpha$ and
which corresponds to the shift in the
regular representation) is complicated.

Some discussions on the spin chain representations and their importance for physics
follows in the Conclusion.

\subsection{The Generalized Jordan-Wigner Transformation}\label{jordanwigner}

We have already mentioned the representation of every algebra
$\A$ as a $C^*$ algebra $\Pi(\A)\subseteq\B(\ell_2(G))$.
We can give another representation that can be considered as a  generalization of the
Jordan-Wigner transformations, that relates the spin lattice with the Fermions on a lattice.
We represent $\A$ as a
subalgebra of the doubled spin chain
$\bigotimes _{m=-\infty }^{\infty}(M_{d\times d}\otimes M_{d\times d})_m$
$=\bigotimes _{n=-\infty }^{\infty}(M_{d\times d})_n$,
where $W^{(0)}_{\vec{k}}$ with
$\vec{k}=(k_1,k_2)$ is identified with the
infinite tensor product
\beq
\label{tens-prod}
\left(\bigotimes_{n=1}^{+\infty}(W_{0,b_n})_{-2n}\otimes
(W_{0,a_n})_{-2n+1}\right)\,
\otimes\,
(W_{k_1,0})_0\otimes (W_{k_2,k_1})_1\,\bigotimes_{n=2}^{+\infty}(1)_n\ ,
\eeq
where the Weyl operator $W_{k_1,0}$ is at site
$n=0$, while $W_{k_2,k_1}$ is at site $n=1$,
whereas the operators $W_{0,b_n}$ are located at sites $-2n$
and $W_{0,a_n}$ at sites $-2n+1$.
Moreover,  the components $a_n$ and $b_n$ are determined by $\vec k$
and the commutation relations via
$$
A_n\left (\ba{c} k_1 \\ k_2 \ea \right )=\left (\ba{c} b_n \\
a_n \ea \right )\ .
$$
Furthermore, the action of the shift automorphism $\alpha$ is now
represented as a $2$-step translation along the lattice.

Notice that, since the contributions from the infinite tails commute
with each other by construction, finite tensor products of the form
$W^{(0)}_{k_1,k_2}\cdots W^{(N)}_{l_1,l_2}$ may be effectively represented as
elements of the matrix algebra $\bigotimes_{n=0}^{2N}(M_{d\times d})_n$.
It thus follows that the commutant consists of operators of the form
$$
\bigotimes_{n=2j-1}^{-\infty}(1)_n\otimes
(W_{\ell_1,\ell_2})_{2j}\otimes (W_{\ell_2,0})_{2j+1}\otimes\left(
\bigotimes_{k=j+1}(W_{0,b_{k-j}(\ell)})_{2k}\otimes
(W_{0,a_{k-j}(\ell)})_{2k+1}\right),
$$
plus operators belonging to the center. (The center becomes
trivial if (\ref{cond-det}) holds. But this condition is not needed here.)

We can consider the operators to act on the vector $|\Omega>=\otimes_{-\infty<k<\infty} |0>$ where at each lattice point $W_{0,k}|0>=|0>.$ Representing $\widetilde{W}^0_{k_1,k_2}$ by $$\left(\bigotimes_{n=1}^{+\infty}1\otimes 1\right)\,
\otimes\,
(W_{k_1,0})_0\otimes (W_{k_2,})_1\,\bigotimes_{n=2}^{+\infty}(1)_n\ $$ and letting it acting on $|\Omega>=\vert\widetilde{W}_{I_0}\rangle$ we reproduce $l^2(G)$ and therefore the regular representation.

\medskip

\noindent
{\bf Example}\quad
As a concrete illustration of the algebraic construction of above,
let us consider the Price-Powers shift.
This corresponds to $d=2$, and to a generating sequence of matrices
which are either $A_n=\pmatrix{1&0\cr 0&1}$
corresponding to $g(n)=0$ or $A_n=\pmatrix{0&1\cr
-1&0}$
corresponding to $g(n)=1$.
Then, we set $e_0=W^{(0)}_{1,0}$ and observe that operators at odd
places have the form $W_{0,k}$ and therefore all commute.
Consequently, we can remove them from the tensor product
(\ref{tens-prod})
and represent
$$
e_m=\bigotimes_{n=1}^{+\infty}(W_{0,b_n})_{k-n}\otimes
(W_{1,0})_k\bigotimes_{n=k+1}^{+\infty}(1)_n\ .
$$

\subsection{Preliminary Remarks on Invariant States}\label{prelim}

We now turn to the problem of finding the invariant states under the
$2$-step shift.
If we are only interested in the local effects of such automorphism, we
can take the periodic shift in  $\bigotimes_{n=0}^{2N}(M_{d\times d})$,
which is unitarily implemented.

The algebra created by finitely
many Weyl operators is imbedded in $\bigotimes_{n=0}^{2N}(M_{d\times d})$,
so that we can conclude that the automorphism $\alpha:\A\mapsto\A$
is unitarily implemented. Therefore, we can construct states on the
local algebra defined by density operators that commute with the
unitary that implements the periodic shift.
However, in general, these density operators
will not have a limit when $N\rightarrow \infty$.

Another possibility to construct $\alpha$-invariant states is to
start from vectors in the infinite tensor product $\bigotimes_{n=-\infty}^{\infty}\vert\psi_n\rangle$
that are invariant under the shift. The simple choice where all
$\vert\psi_n\rangle$ are identical to an eigenvector
$\vert\phi\rangle$ of $W_{(0,1)}$
gives, independently of the sequence $\{A_n\}_{n\in\mathbb{N}}$,
already the tracial state; indeed, since for every $W_I$ at some
position
$\langle \phi|W_{k,l}|\phi \rangle =0$.
Therefore, the representation is isomorphic to the regular
representation.

If we choose some other vector and assume that $A_n\neq 1$
for infinitely many $n$, again we obtain the tracial state; in fact,
either
$\displaystyle\left|
\langle\psi|W_{0,b_n}|\psi \rangle\,\langle\psi\vert W_{0,a_n}
\vert\psi\rangle\right|\,<\,1$,
infinitely often so that
$$
\prod_{n=-\infty}^{+\infty}\left|
\langle\psi|W_{0,b_n}|\psi \rangle\,\langle\psi\vert W_{0,a_n}
\vert\psi\rangle\right|=0\ .
$$

If for some $N>0$ $A_n=0$ for all $n>N$, we can choose a vector $\phi $ over
$\bigotimes _{n=1} ^{2N}M_{d\times d}^{(n)}$ that is appropriately
entangled over the lattice points to guarantee that
$\langle \phi |W_{k_1,0)}|\phi \rangle \neq 0$.
By averaging this vector over the period $N$, one gets an
expectation value still $\neq 0$.
However, in general, one expects that it decreases with $N$.
Therefore, if for every $N$ we can find $n>N$ such that $A_n\neq 0$ in order to obtain another invariant state in the limit
$N\rightarrow \infty$, it is necessary  to have
correlations between infinitely many
lattice points; this can hardly be satisfied because of
monogamy of entanglement.
Though we are unable to exclude that other invariant states
might be constructed, our considerations already indicate that the tracial
state will turn out to be the only invariant state under
appropriate conditions on the defining sequence of matrices $A_n$.

\section{Invariant states}\label{invariantstates}

Given $(\A, \alpha)$, let $\omega$ be an invariant state such that
$\omega \circ \alpha=\omega$ and consider the corresponding GNS
representation $\pi_\omega$ of $\A$ as a $C^*$ algebra of bounded
operators on the GNS Hilbert space $\Ha_\omega$ with
cyclic vector $\vert\Omega\rangle$.
Namely, $\omega(A)=\langle\Omega\vert\pi_\omega(A)\vert\Omega\rangle$;
further, the shift automorphism $\alpha $ is implemented by a unitary
$U_{\omega}$ such that
$U_\omega\vert\Omega\rangle=\vert\Omega\rangle$.
The following simple lemmas hold.
\medskip

\begin{lemma}
\label{lemma1}
Let $P^\omega_0$ denote the projection onto the $U_\omega$-invariant
subspace of $\Ha_\omega$.
If
\beq
\label{maintool3}
\langle\Omega\vert\pi _{\omega}(W^*_I)P^{\omega }_0 \pi _{\omega
}(W_I)\vert\Omega \rangle=0
\eeq
is true for all $W_I$, then $\omega$ is tracial that is
$\omega(XY)=\omega(YX)$ for all $X,Y\in\A$.
\end{lemma}

\noindent
\begin{proof}
Since $P^\omega_0\geq\vert\Omega\rangle\langle\Omega\vert$, the
assumption implies
$$
|\omega(W_I)|^2=\langle\Omega\vert\pi _{\omega}(W^*_I)
\vert\Omega\rangle\langle\Omega\vert\pi _{\omega}(W_I)\vert\Omega
\rangle
\leq
\langle\Omega\vert\pi _{\omega}(W^*_I)P^{\omega }_0 \pi _{\omega
}(W_I)\vert\Omega \rangle=0\ ,
$$
whence, by Cauchy-Schwartz, $\omega(W_I\,W_J)\neq 0$ only if
$W_IW_J=1$. Thus, $\omega(W_I\,W_J)=\omega(W_J\,W_I)$ for all $W_{I,J}$, whence
$\omega(XY)=\omega(YX)$ for $\A$ is generated by linear combinations
of $W_I$s.
\end{proof}

\begin{lemma}
\label{lemma2}
If $\omega$ is $\alpha$-invariant, then, setting $u_n(I):=u_n(I;I)$
in~(\ref{comrel}),
\beq
\label{maintool1}
\langle\Omega \vert\pi _{\omega }(W^*_I)P^{\omega } _0\pi _{\omega }(W_I)
\vert\Omega\rangle
=\lim _{N\rightarrow \infty }\frac{1}{N}\sum^{N-1} _{n=0}
e^{2\pi\,i\,u_n(I)}\langle\Omega\vert\pi _{\omega
}(W_I)U^{-n}_{\omega } \pi _{\omega }(W^*_I)\vert\Omega\rangle\ .
\eeq
\end{lemma}

\begin{proof}
The mean ergodic theorem of von Neumann~(\cite{Quef}) and~(\ref{comrel}) imply
\beqan
\langle\Omega \vert\pi _{\omega }(W^*_I)P^{\omega } _0\pi _{\omega }(W_I)
\vert\Omega\rangle&=&
\lim _{N\rightarrow \infty }\frac{1}{N}\sum ^{N-1} _{n=0}\omega\Bigl(
W^*_I\alpha ^n\Bigl(W_I\Bigr)\Bigr)\\
&=&
\lim _{N\rightarrow \infty }\frac{1}{N}\sum ^{N-1} _{n=0}
e^{2\pi\,i\,u_n(I)}\,
\omega\Bigl(W_I\alpha ^{-n}\Bigl(W^*_I\Bigr)\Bigr)\ .
\eeqan
\end{proof}

\noindent
{\bf Remarks 1}

\begin{enumerate}
\item
If $W_I=1$, then
$\langle\Omega\vert
\pi_\omega(W_I^*)P^\omega_0\pi_\omega(W_I)
\vert\Omega\rangle=1$,
$u_n(I)=0$ for all $n\in\mathbb{N}$ and (\ref{maintool1})
is trivially satisfied.
\item
Using the spectral decomposition
$
U_\omega=\int_{Sp(U_\omega)}{\rm
d}P_\omega^\lambda\, {\rm e}^{2\pi\,i\lambda}$,
(\ref{maintool1}) reads
\beq
\label{maintool1a}
\langle\Omega \vert\pi _{\omega }(W^*_I)P^{\omega } _0\pi _{\omega }(W_I)
\vert\Omega\rangle
=\lim _{N\rightarrow \infty }\int_{Sp(U_\omega)}{\rm
d}\mu_I(\lambda)
\frac{1}{N}\sum^{N-1} _{n=0}
{\rm e}^{2\pi\,i(u_n(I)-n\lambda)}\ ,
\eeq
where
${\rm d}\mu_I(\lambda):={\rm d}(\langle\Omega\vert\pi _{\omega
}(W_I)P^\lambda_\omega\pi _{\omega }(W^*_I)\vert\Omega\rangle$.
\end{enumerate}
\medskip

We now concentrate on the sequences
$u(I):=\{u_n(I)\}_{n\in\mathbb{N}}$ and
$u^\lambda_n(I):=\{u_n(I)-\lambda\,n\}_{n\in\mathbb{N}}$ and
study their spectrum~\cite{Quef}.
In order to properly introduce this notion, consider a
sequence $v=\{v_n\}_{n\in\mathbb{N}}$ taking its values
$v_n\in\mathbb{C}$ in a compact
subset of the complex numbers. For all
$k\in\mathbb{N}$ the partial sums
\beq
\label{Funct}
S_N(k):=\frac{1}{N}\sum_{n=0}^{N-1}v_n^*v_{n+k}
\eeq
are bounded, thus, the sequence $S(k):=\{S_N(k)\}_{N\in\mathbb{N}}$ has
accumulation points and, by a Cantor-like diagonalization argument
(for details see the Appendix),
there exists at least one subsequence $\{N_j\}_{j\in\mathbb{N}}$ such that
the limit
$$
s_k(v):=\lim_{j\to\infty}\frac{1}{N_j}\sum_{n=0}^{N_j-1}v_n^*v_{n+k}
$$
exists for all $k\in\mathbb{N}$.
By setting $s_{-k}:=s_k^*$, one obtains a positive-definite sequence
(details are again in the Appendix), that is a sequence
$s(v)=\{s_k(v)\}_{k\in\mathbb{Z}}$ such that
$$
\sum_{i,j}z_i^*s_{i-j}(v)z_j\geq 0
$$
for all sequences
$\{z_i\}_{i\in\mathbb{Z}}$ such that
$\sum_{i\in\mathbb{Z}}|z_i|^2<\infty$.
Then, by Bochner's theorem,
$$
s_k(v)=\int_{0}^{1}{\rm d}\mu_v(x)\,{\rm e}^{2\pi\,ikx}\ ,
$$
where ${\rm d}\mu_v(x)$ is a positive (\textit{correlation}) measure
on $[0,1)$ such that
$$
\int_{0}^{1}{\rm d}\mu_v(x)=\lim_{j\to\infty}\frac{1}{N_j}
\sum_{n=0}^{N_j-1}|v_n|^2\ .
$$
If the correlation measure of a sequence $v$ is the Lebesgue measure,
then $s_k(v)=0$
whenever $k\neq0$ and the sequence $v$ is said to be
\textit{uniformly distributed}.
Therefore, it makes sense to introduce the

\begin{definition}[Spectrum of a sequence]\cite{Quef,DT}
\label{def-spec}
Given a sequence $v:=\{v_n\}_{n\in \mathbb{N}}$ with values in
a compact subspace of $\mathbb{C}$, its Fourier-Bohr spectrum is
given by
$$
{\rm Sp}(v):=\left\{\mu\in[0,1)\,:\,
\limsup_{N\to\infty}\frac{1}{N}\,
\left\vert\sum_{n=0}^{N-1}v_n\,{\rm e}^{-2\pi\,in\mu}
\right\vert\neq0\right\}\ .
$$
\end{definition}
\medskip

In other words, the spectrum of a sequence $v$ is the subset of
values $\mu\in[0,2\pi)$
such that the sequences
$v(\mu)=\{v_n\exp(-2\pi in\mu)\}_{n\in\mathbb{N}}$
are not uniformly distributed.
Equivalently, $\mu\notin Sp(v)$ if and only if
\beq
\label{maintool4}
\lim_{j\to\infty}\frac{1}{N_j}\sum_{n=0}^{N_j-1}v_n\,{\rm e}^{-2\pi\,i\mu\,n}=0
\eeq
for all converging subsequences of partial sums.

By means of the spectral properties of sequences, we can derive
sufficient conditions that force the invariant state $\omega$ to be
tracial.
\medskip

\begin{lemma}
\label{lemma3}
Let $v(I):=\{{\rm e}^{2\pi i\,u_n(I)}\}_{n\geq 0}$; then,
the dynamical system $(\A,\alpha)$ has the tracial state as its only
invariant state if for each $I$,
$Sp(v(I))$ is either $\emptyset$ or
$\{0\}$ with
$$
\lim_{N\to\infty}\frac{1}{N}\sum_{n=0}^{N-1}v_n(I)
=\sum_{j=0}^{d-1}p_j(I)e^{\frac{2\pi
i}{d}j}\Bigl(=:\nu(v(I))\Bigr)\qquad(*)
$$
for some $p_j(I) \geq 0$, $j\in D$, $p_0 \neq 1$,
$\sum_{j=0}^{d-1}p_j(I)=1$.
\end{lemma}
\medskip

\noindent
\begin{proof}
If $Sp(v(I))=\emptyset$, then, using~(\ref{maintool4}) and
dominated convergence, the right
hand side of~(\ref{maintool1a}) vanishes and the result follows from
Lemmas~\ref{lemma1} and~\ref{lemma2}.

If for some $I$ $Sp(v(I))=\{0\}$ and relation (*) holds for such $I$, then, using
again~(\ref{maintool4}) and dominated convergence,
\beq
\label{maintool5}
\langle\Omega\vert\pi _{\omega }(W^*_I)P^{\omega }_0
\pi_{\omega }(W_I)\vert\Omega\rangle=\nu(v(I))
\langle\Omega\vert
\pi _{\omega }(W_I)P^{\omega}_0\pi _{\omega
}(W^*_I)\vert\Omega\rangle\ .
\eeq
Since $|\nu(v(I))|<1$, by exchanging $W_I$ and $W_I^*$, one gets
$$
\langle\Omega\vert\pi _{\omega}(W^*_I)P^{\omega } _0
\pi _{\omega }(W_I)\vert\Omega\rangle<
\langle\Omega\vert\pi _{\omega}(W_I)P^{\omega } _0
\pi _{\omega }(W^*_I)\vert\Omega\rangle<
\langle\Omega\vert\pi _{\omega}(W^*_I)P^{\omega } _0
\pi _{\omega }(W_I)\vert\Omega\rangle\ .
$$
Thus, $\langle\Omega\vert\pi _{\omega }(W^*_I)P^{\omega}_0\pi_{\omega
}(W_I)\vert\Omega\rangle=0$ and
Lemmas~\ref{lemma1} and~\ref{lemma2} apply.
\end{proof}

For some given $I$, for instance a singleton $I=\{1\}$, there surely
exist sequences of matrices $\{A_n\}_{n\in\mathbb{N}}$,
with entries from $\{0,1,\ldots,d-1\}$, such that
$Sp(v(I))=\emptyset$, or $Sp(v(I))=\{0\}$.

However, in order to use the previous Lemma, we
have to make sure that there exist sequences $\{A_n\}_{n\in\mathbb{N}}$ such
that conditions $1$ or $2$ in the previous Lemma are fulfilled
for all $I$.

In the following, we shall consider the $4$ entries $a_{ij}(n)$
of the matrices $A_n$ as random processes with values
from $\{0,1,\ldots\,d-1\}$.
Then, we shall focus upon  the space $\X$ of sequences
$\vec{x}=\{\vec{x}_n\}_{n\in\mathbb{N}}$, where
$\vec{x}_n=(a_{11}(n),a_{12}(n),
a_{21}(n),a_{22}(n))$ are
$4$-valued vectors with the entries of the matrices $A_n$ as
components. If we want in addition to meet the requirements in (2.2), e.g. that $Det(A_{n,n})=1 \quad \forall n,$ we can restrict to $2$ entries $(a_{11}(n),a_{12}(n))$ but still keep enough randomness.

We equip $\X$ with the shift-automorphism
$(\sigma(\vec{x}))_n=\vec{x}_{n+1}$ and with a $\sigma$-invariant
measure $\mu$ (defined on the $\sigma$-algebra of cylinders).
Concretely, if $f$ is a measurable function on $\X$, then its mean
value with respect to $\mu$ is given by
$$
\mu(f)=\int_\X{\rm d}\mu(\vec{x})\,f(\vec{x})\ ;
$$
furthermore, $\mu(f\circ\sigma)=\mu(f)$.

We shall assume that the dynamical system $(\X,\mu,\sigma)$ is mixing
(hence ergodic), that is, if $f$ and $g$ are two essentially bounded
functions on $\X$ with respect to $\mu$, then
\beq
\label{mixing}
\lim_{k\to+\infty}\mu(f\cdot(g\circ\sigma^k))=\mu(f)\,\mu(g)\ ,
\eeq
We shall call a sequence $\vec{x}\in\X$ typical if for all measurable functions it is
self-averaging and mixing in the following sense:
\beqa
\label{cond1}
\lim_{N\to+\infty}\frac{1}{N}\sum_{n=0}^{N-1}f\circ\sigma^n(\vec{x})
&=&\mu(f)\\
\label{cond2}
\lim_{k\to+\infty}
\lim_{N\to+\infty}\frac{1}{N}\sum_{n=0}^{N-1}f\circ\sigma^n(\vec{x})\,
g\circ\sigma^{n+k}(\vec{x})&=&\mu(f)\,\mu(g)\ .
\eeqa
As a concrete example, consider $(\X,\mu,\sigma)$ as
the product of $4$ identically
distributed independent Bernoulli processes.
\medskip

\begin{theorem}
\label{Th1}
Let $\{A_n\}_{n\in\mathbb{N}}$ be a sequence of matrices which provides
a typical sequence $\vec{x}\in\X$ with respect to a shift-invariant measure $\mu$
as explained above.
Let $u(I)=\{u_n(I)\}_{n\in\mathbb{N}}$ be the sequences defined in Lemma \ref{lemma2};
then, for all $I$, the spectrum of $u(I)$ is $\emptyset$ or $\{0\}$.
\end{theorem}
\medskip

\noindent
{\bf Proof:}\quad
Observe that the quantities
$\displaystyle{\rm e}^{-2\pi\,i\,u_n(I)}$ can be regarded as
measurable functions on $\X$; more precisely, let $P_j$ project out of
the sequence $\vec{x}$ the $j$-th component ($P_j\vec{x}=\vec{x}_j$).
Then, consider the expression of $u_n(I)$ given by (\ref{comrel}); it
turns out that one can write
\beqan
{\rm e}^{2\pi\,i\,u_n(I)}&=&\prod_{a=1}^{n_I}\prod_{b=1}^{n_I}
G_{\vec{k}_a\vec{k}_b}\circ\sigma^{b-a}\circ
P_n(\vec{x})\quad\hbox{where}\\
G_{\vec{k}_a\vec{k}_b}\circ\sigma^{b-a}\circ
P_n(\vec{x})&=&{\rm e}^{\frac{2\pi\,i}{d}
\,\sigma(\vec{k}_a,A_{b-a+n}\vec{k}_b)}\ .
\eeqan
Therefore, the function on the left hand side of the first equality
belongs to the class of function on $\X$ that have been used to
define typical sequences (see conditions (\ref{cond1}) and (\ref{cond2})).

Now, we consider the following quantity
\beqan
&&
(*):=\lim _{K\rightarrow \infty } \lim _{N\rightarrow \infty} \frac{1}{NK}\sum_{k=0}^{K-1}\sum_{n=0}^{N-1}{\rm e}^{-2\pi\,i\,
(u_n(I)+\lambda\,n)}\,{\rm e}^{2\pi\,i(u_{n+k}(I)+\lambda(n+k))}=\\
&&
\hskip 1cm
=\lim _{K\rightarrow \infty } \lim _{N\rightarrow \infty}\frac{1}{K}\sum_{k=0}^{K-1}{\rm e}^{2\pi\,i\,\lambda\,k}\,
\frac{1}{N}
\sum_{n=0}^{N-1}{\rm e}^{2\pi\,i(u_{n+k}(I)-u_n(I))}\ .
\eeqan
Notice that the second average is just the function $S_N(k)$ in
(\ref{Funct}); then, using
the conditions defining typical sequences, one gets
$$
\lim_{k\to+\infty}\lim_{N\to+\infty}S_N(k)=\left|\mu\left(
\prod_{a=1}^{n_I}\prod_{b=1}^{n_I}
G_{\vec{k}_a\vec{k}_b}\circ\sigma^{b-a}\right)\right|^2=(**)\ .
$$
Therefore, one concludes that in the limits $N\to+\infty$ and
$K\to+\infty$, the expression $(*)$
converges to $0$ if $\lambda\neq 0$, or to $(**)$ if $\lambda=0$.
Thus, the spectrum of $u_n(I)$ is either empty or
consists of $\lambda=0$, because for a typical sequence the expression $(*)$ equals
$$
 \lim _{N\rightarrow \infty}\left|\frac{1}{N}\sum_{n=0}^{N-1}{\rm e}^{2\pi\,i\,
(u_n(I)+\lambda\,n)}\right|^2\ .
$$
\begin{flushright} $\square$
\end{flushright}
\medskip

Finally, by applying Lemma~\ref{lemma3}, we obtain
\medskip

\begin{theorem}
Given a sequence of $2\times 2$ matrices $\{A_n\}_{n\in\mathbb{N}}$
with integer entries such that it is typical in the sense of
conditions (\ref{cond1}) and (\ref{cond2}), then only the tracial state
on the twisted Weyl algebra $\A$ is $\alpha$-invariant.
\end{theorem}

\section{Conclusions}

We have constructed quantum algebras and automorphisms on them such
that they permit only one invariant state, namely the tracial state.
The main tool in the proof was the application of results on classical
random sequences. We should notice however, that our proof can also be
applied to the Price-Powers shift, but there it does not cover all
possibilities for which the result holds.
Indeed, also bitstreams not fulfilling conditions (\ref{cond1}) and
(\ref{cond2}) may have the tracial state as the only shift-invariant state.
Indeed, these two conditions are only sufficient, but not necessary.

Also, the failure in section \ref{jordanwigner} to construct invariant states does not
refer to the special construction of the automorphism, rather it is
long-range non-commutativity which seems to be important. If we can represent the operators as in (\ref{weyl-spin1}) we can embody this non-commutativity in \beq
\label{non-comm}
[[\cdots[\alpha^n(A)\,,\,B_0],B_1],\cdots\,B_n]
\neq 0
\eeq for an appropriate sequence of $B_k$ where the $B_k$ is localized at the lattice point $k$. Therefore the operator $\alpha ^n(A)$ is not only spread as it happens for quasifree evolution but got delocalized also in a multiplicative sense. Nevertheless it has still finite velocity in the sense that, at
every step, an operator in the local algebra $\bigotimes_{n=0}^N(
M_{d\times d})_n$ is mapped into an operator located in the algebra
$\bigotimes_{n=0}^{N+1}(M_{d\times d})_n$.

We expect that the occurrence of non-trivial multicommutators
that do not vanish should be typical of interacting quantum systems.
Of course, in general, the dynamics is such that one deals with
continuous automorphism groups and with multicommutators
by far more complicated. However, the preceding analysis indicates
that in the present abstract model multicommutators are responsible for the
nonexistence of invariant states. This gives a hint that also in more
general situations multicommutators should play an important
role in the search for invariant states.

\bigskip

\section{Appendix}
\subsubsection*{The Cantor-like argument needed, following equation (\ref{Funct})}

Define $S_N$ as a function on the integers such that
$$
S_N(k):=\frac{1}{N}\sum_{n=0}^{N-1}v_n^*v_{n+k}
$$
 as in (\ref{Funct}), where
 $|S_N(k)|\leq K$ by assumption.
Now
${S_N(1)}$ need not converge, but one can extract
a subset $\{N(j,1),\,j=1,2,3\ldots\}\in\mathbb N$
such that ${S_{N(j,1)}(1)}$ converges, as $j\to\infty$.
Then we proceed inductively, define for each $k$ a smaller subset
$\{N(j,k),\,j=1,2,3\ldots\}\subset\{N(j,k-1)\}$, such that ${S_{N(j,k)}(k)}$ converges as $j\to\infty$,
with $k$ fixed.

Now, and here is the Cantor diagonalization trick,
one considers the set $\{N_j:=N(j,j)\}$.
The sequences ${S_{N(j,j)}(k)}$ converge, as $j\to\infty$, for each $k$.

\subsubsection*{Positive-definiteness}

Consider a set $\{z_i\},\,|i|\leq M$, extend the set $\{v_n\}$, defining $v_n:=0$ for $n<0$,
and transform
$$\sum_{i,j}z_i^*s_{i-j}z_j=\lim_{N\to\infty}\frac1N\sum_{i,j}\sum_{m=i}^{N+1+i}v^*_{m-i}z_i^*v_{m-j}z_j\,.$$
Considering the bounds $|v_n|<V$, $|z_i|<Z$, one gets
$$\sum_{i,j}z_i^*s_{i-j}z_j=
\lim_{N\to\infty}\frac1N \left[\sum_{m\in\mathbb Z}\left(\sum_i v_{m-i}z_i\right)^*\left(\sum_jv_{m-j}z_j\right)
+\textrm{O}\, (M^2\cdot V^2\cdot Z^2)\right].$$
The error-term vanishes in the limit $N\to\infty$ (to be taken over the subset of $N$ where limits of the $S_N$ exist).
Then one may consider approaching $\ell^2$ sequences of $z_i$ by finite sequences.
\bigskip

\noindent
{\bf Acknowledgement}\quad
It is a pleasure to thank Dimitri Petritis, Johannes Schoissengeier
and Valerio Cappellini for useful hints and comments.

\bibliographystyle{plain}

\end{document}